\documentclass{easychair}

\usepackage{bm}
\usepackage{enumerate}

\usepackage{doc}

\usepackage{amsthm}
\newtheorem*{proposition*}{Proposition}

\def\I#1{\rlap{(}[#1\rlap{)}]}
\usepackage{algorithmic,multirow}

\usepackage[ruled,vlined,linesnumbered]{algorithm2e}

\title{Quantifier Elimination for Reasoning in Economics}
\titlerunning{Quantifier Elimination for Reasoning in Economics}

\author{
Casey B. Mulligan\inst{1}
\and
Russell Bradford \inst{2}
\and
James H. Davenport \inst{2}
\and\\
Matthew England \inst{3}
\and 
Zak Tonks \inst{2}
}

\institute{
  University of Chicago, USA\\
  \email{c-mulligan@uchicago.edu}
\and
   University of Bath, U.K.\\
   \email{\{R.Bradford, J.H.Davenport, Z.P.Tonks\}@bath.ac.uk}\\
\and
   Coventry University, U.K\\
   \email{Matthew.England@coventry.ac.uk}\\
 }

\authorrunning{Mulligan et al.}

\begin{document}

\maketitle

\begin{abstract}
We consider the use of Quantifier Elimination (QE) technology for automated reasoning in economics.  QE dates back to Tarski's work in the 1940s with software to perform it dating to the 1970s.  There is a great body of work considering its application in science and engineering but we show here how it can also find application in the social sciences.  We explain how many suggested theorems in economics could either be proven, or even have their hypotheses shown to be inconsistent, automatically; and describe the application of this in both economics education and research.  

We describe a bank of QE examples gathered from economics literature and note the structure of these are, on average, quite different to those occurring in the computer algebra literature.  This leads us to suggest a new incremental QE approach based on result memorization of commonly occurring generic QE results. 
\end{abstract}


\def\query{{\buildrel?\over\Rightarrow}}
{ 
\def\v{{\bf v}}

\section{Introduction}

A general task in economic reasoning is to determine whether, with variables $\v=(v_1,\ldots,v_n)$, the hypotheses $H$ follow from the assumptions $A$, i.e. is it the case that
$
\forall\v A \Rightarrow H
$?

Ideally the answer would be \verb+true+ or \verb+false+, but of course in practice life is more complicated.  There are in fact four possibilities:
\begin{center}
\begin{tabular}{r|cc}
& $\lnot\exists \v[A\land \lnot H]$ & $\exists \v[A\land \lnot H]$ \\ \hline
$\exists \v[A\land H]$     & True          & Mixed \\
$\lnot\exists \v[A\land H]$& Contradictory Assumptions & False
\end{tabular}
\end{center}
Should technology provide one of these automatically then in all cases an economist gains important information: either a proof or a disproof of the theory; an identification of where the theorem may be true (a description of $\{\v:A(\v)\Rightarrow H(\v)\}$); or the identification of contradictory assumptions.

Further than this, the economist could vary the question: the assumptions generating a \verb+true+ result can be weakened, or the assumptions generating a \verb+mixed+ result strengthened, by quantifying more or less of the variables in $\v$.  For example, we might partition $\v$ as $\v_1\cup\v_2$ and ask for $\{\v_1:\forall \v_2 A(\v_1,\v_2)\Rightarrow H(\v_1,\v_2)\}$. The result in these cases is a formula in the free variables that weakens or strengthens the assumptions as appropriate.  Should technology automatically generate these formulas, the economist gains information about how to reformulate assumptions that justify his hypotheses. 

\newpage

Such problems fall within the framework of \emph{Quantifier Elimination} (QE).  QE refers to the generation of an equivalent quantifier free formula from one that contains quantifiers.  QE is known to be possible over real closed fields thanks to the seminal work of Tarski \cite{Tarski1948}.  Practical implementations followed the work of Collins on the Cylindrical Algebraic Decomposition (CAD) method \cite{Collins1975} and Weispfenning on Virtual Substitution \cite{Weispfenning1988}.  There are modern implementations of QE in \textsc{Mathematica} \cite{Strzebonski2006}, \textsc{Redlog} \cite{DS97a}, \textsc{Maple} (\textsc{SyNRAC} \cite{IYA14} and the \textsc{RegularChains}  Library \cite{CM14c}) and \textsc{Qepcad-B} \cite{Brown2003b}.

QE has a long history within which it has found many applications within engineering and the life sciences.  Some recent examples include 
the derivation of optimal numerical schemes \cite{EH16},
artificial intelligence to pass university entrance exams \cite{WMTA16},
weight minimisation for truss design \cite{CC17},
and steady state analysis of biological networks \cite{BDEEGGHKRSW17, EEGRSW17}.  The recent survey article \cite{Sturm2017} has applications in geometric theorem proving, verification and the life sciences.

However, applications in the social sciences are lacking in the QE literature\footnote{The nearest we can find is \cite{LiWang2014a}, though there is also the thread of work in \cite{Caminatietal2015a}.}.  Indeed, on a few occasions when QE algorithms have been mentioned in economics they have been characterized as too "computationally demanding" \cite{Carvajal2014} and "something that is do-able in principle, but not by any computer that you and I are ever likely to see" \cite{Steinhorn2008}.  But these assertions are based on interpretations of theoretical computer science results rather than experience with actual software applied to an actual economic reasoning problem.  This paper examines real examples that introduce economics as a new, potentially large, application area for QE.

\vspace*{0.1in}

The paper proceeds as follows.  In Section \ref{SEC:Examples} we describe in detail some examples from economics, ranging from textbook examples common in education to questions arising from current research discussions.  We explain how they may be resolved via QE and then in Section \ref{SEC:ExampleCollection} we describe a collection of similar examples and analyse the structure in comparison to those from the QE literature, concluding that they are not well represented.   In Section \ref{SEC:BlockQE} we note that while existing software can already solve the examples QE methods could be adapted to better deal with them by incrementally eliminating quantifiers using memorised generic QE results as building blocks.  In Section \ref{SEC:BuildBlock} we examine how some such blocks can be generated with minimal applications of traditional QE techniques.


\section{Economic reasoning as a case of Tarski's elementary algebra}
\label{SEC:Examples}

The fields of economics ranging from macroeconomics to industrial organization to labor economics to econometrics involve deducing conclusions from assumptions or observations.  Will a corporate tax cut cause workers to get paid more?  Will a wage regulation reduce income inequality?  Under what conditions will political candidates cater to the median voter?

\subsection{Comparative static analysis}

We start with Alfred Marshall's \cite{Marshall1895} classic, and admittedly simple, analysis of the consequences of cost-reducing progress for activity in an industry.  Marshall  concluded that, for any supply-demand equilibrium in which the two curves have their usual slopes, a downward supply shift increases the equilibrium quantity $q$ and decreases the equilibrium price $p$.  

One way to express his reasoning in Tarski's framework is to look at the industry's local comparative statics: meaning a comparison of two different equilibrium states between supply and demand.  
With a downward supply shift represented as $da > 0$ we have here:
\begin{align*}
&A \equiv D'(q)<0 \wedge S'(q)>0 \\
&\qquad \wedge \dfrac{d}{da}\big(S(q)-a\big)=\frac{dp}{da}
\wedge \frac{dp}{da}=\frac{d}{da}D(q)
\\
&H \equiv \frac{dq}{da}>0 \wedge \frac{dp}{da}<0
\end{align*}
where: 
\begin{itemize}
\item $D'(q)$ is the slope of the inverse demand curve in the neighborhood of the industry equilibrium;
\item $S'(q)$ is the slope of the inverse supply curve; 
\item $\dfrac{dq}{da}$ is the quantity impact of the cost reduction; and 
\item $\dfrac{dp}{da}$ is the price impact of the cost reduction.
\end{itemize}
Economically, the first atoms of $A$ are the usual slope conditions: that demand slopes down and supply slopes up.  The last two atoms of $A$ say that the cost change moves the industry from one equilibrium to another.  Marshall's hypothesis was that the result is a greater quantity traded at a lesser price.

These``variables'' $\v = \left\{D'(q),S'(q),\dfrac{dq}{da},\dfrac{dp}{da}\right\}$ are four real numbers $(v_1,v_2,v_3,v_4)$, and so, after applying the chain rule, $A$ and $H$ may be understood as Boolean combinations of polynomial equalities and inequalities:
\begin{align*}
&A \equiv v_1<0 \wedge v_2>0  \wedge v_3v_2-1=v_4 \wedge v_4=v_3v_1 
\\
&H \equiv v_3>0 \wedge v_4<0
\end{align*}
 Thus Marshall's reasoning fits in the Tarski framework and therefore is amenable to QE tools, and any of the modern QE implementations mentioned in the introduction can conclude instantly that $\forall\v A \Rightarrow H$ is \texttt{true} and thus confirm Marshall's conclusion.

\subsection{Scenario analysis}
\label{SUBSEC:Scenario}

Economics is replete with \emph{"what if?"} questions.  Such questions are logically and algebraically more complicated, and thereby susceptible to human error, because they require tracking various scenarios.  
Writing at nytimes.com\footnote{http://krugman.blogs.nytimes.com/2012/11/03/soup-kitchens-caused-the-great-depression/
}, Economics Nobel laureate Paul Krugman asserted that whenever taxes on labor supply are primarily responsible for a recession, then wages increase.  
Two scenarios are discussed here: what actually happens ($act$) when taxes ($t$) and demand forces ($a$) together create a recession, and what would have happened ($hyp$) if taxes on labor supply had been the only factor affecting the labor market.  Expressed logically we have:
\begin{align*}
A &\equiv \bigg( \, \frac{\partial D(w,a)}{\partial w}<0 \wedge \frac{\partial S(w,t)}{\partial w}>0  
\\
&\qquad \wedge
	\frac{\partial D(w,a)}{\partial a}=1 \wedge
    \frac{\partial S(w,t)}{\partial t}=1 
\\
&\qquad \wedge
    \frac{d}{dact}\big(D(w,a)=q=S(w,t)\big)
\\
&\qquad  \wedge   
    \frac{d}{dhyp}\big(D(w,a)=q=S(w,t)\big)
\\
&\qquad \wedge
	\frac{d t}{dact}=\frac{d t}{dhyp}\wedge
	\frac{d a}{dhyp}=0
\\
&\qquad \wedge
	\frac{dq}{dhyp}<\frac{1}{2} \frac{dq}{dact}<0 \, \bigg)
\\
H &\equiv \frac{dw}{dact}>0
\end{align*}
In Economics terms, the first line of assumptions contains the usual slope restrictions on the supply and demand curves.  Because nothing is asserted about the units of $a$ or $t$, the next line just contains normalizations.  The third and fourth lines say that each scenario involves moving the labor market from one equilibrium (at the beginning of a recession) to another (at the end of a recession).  The fifth line defines the scenarios: both have the same tax change but only the $act$ scenario has a demand shift.  The final assumption / assertion is that a majority of the reduction in the quantity $q$ of labor was due to supply (that is, most of $\frac{dq}{dact}$ would have occurred without any shift in demand).  The hypothesis is that wages $w$ are higher at the end of the recession than they were at the beginning.

\noindent Viewed as a Tarski formula this has twelve variables 
\begin{align*}
\v = \bigg\{
&\frac{da}{dact},\frac{da}{dhyp},\frac{dt}{dact},\frac{dt}{dhyp}, 
\\
&\frac{dq}{dact},\frac{dq}{dhyp},\frac{dw}{dact},\frac{dw}{dhyp}, 
\\
&\frac{\partial D(w,a)}{\partial a},\frac{\partial S(w,t)}{\partial t},\frac{\partial D(w,a)}{\partial w},\frac{\partial S(w,t)}{\partial w}\bigg\},
\end{align*}
each of which is a real number representing a partial derivative describing the supply and demand function or a total derivative indicating a change over time within a scenario.   QE shows that the result is \verb+mixed+: even when all of the assumptions are satisfied, it is possible that wages actually go down.  Moreover, if $\frac{\partial D(w,a)}{\partial w}$ and $\frac{\partial S(w,t)}{\partial w}$ are left as free variables, the resulting formula shows that one also needs to assume that labor supply is at least as sensitive to wages as labor demand is.\footnote{The quantifier-free formula is $\frac{\partial S(w,t)}{\partial w}\geq-\frac{\partial D(w,a)}{\partial w}$>0 (recall that the demand slope is a negative number).  See also \cite{Mulligan2012}.}

\subsection{Vector summaries}

Economics problems sometimes involve an unspecified (and presumably large) number of variables.  Take Sir John Hicks' analysis of household decisions among a large number $N$ of goods and services \cite{Hicks1946}.  The quantities purchased are represented as a $N$-dimensional vector, as are the prices paid.  When prices are $\bf p$ ($\hat{\bf p}$), the household makes purchases $\bf q$ ($\hat{\bf q}$), respectively:\footnote{The price change is compensated in the Hicksian sense, which means that $\bf q$ and $\hat{\bf q}$ are the same in terms of results or "utility", so that the consumer compares them only on the basis of what they cost (dot product with prices).  }
\begin{equation*}
A \equiv (\bf p \cdot \bf q \leq \bf p \cdot \hat{\bf q}) 
\, \wedge \,  
(\hat{\bf p} \cdot \hat{\bf q} \leq \hat{\bf p} \cdot \bf q).
\end{equation*}
Hicks asserted that the quantity impact of the price changes $ \hat{\bf q} - \bf q$ cannot be positively correlated with the price changes $ \hat{\bf p} - \bf p$:
\begin{equation*}
	H \equiv (\hat{\bf q} - \bf q) \cdot (\hat{\bf p} - \bf p) \leq 0
\end{equation*}
Although the length of the vectors is unspecified, Hicks' reasoning depends only on the vector dot products, four of which appear above.  Hicks implicitly assumed that prices and quantities are real valued, which places additional restrictions on the dot products.  These restrictions are not needed in this instance, but if they were we could add to the list of variables to reflect all ten products that are possible with four vectors and restrict their Gram matrix to be positive semidefinite.   
QE technology reveals that no counterexample exists: that Sir John Hicks was correct.


\section{The algebraic structure of economic reasoning problems}
\label{SEC:ExampleCollection}

We have assembled a new benchmark set of 45 economic theorems, chosen for their importance in economic reasoning and their consistency with the Tarski framework, but not on the basis of their complexity.  The set is freely available from the following URL.\\
\texttt{\url{https://doi.org/10.5281/zenodo.1226892}}.

The sentences tend to be significantly more complicated than the examples described above: on average, each contain 17 variables and 19 polynomials.  Examples with this number of variables rarely appear in the QE literature, especially because algorithms like CAD have complexity doubly exponential in the number of variables.\footnote{QE itself is known to be doubly exponential in the number of quantifier changes.}

However, there are other features of these examples which give hope for their resolution with QE technology.  

\subsection{Low degree in individual variables}

The maximum total degree of the polynomials in each sentence ranges but averages more than four.  However, the maximum degree in any one variable is usually two and not more than three.  It is the degree of an individual variable rather than the total degree that is of most importance, since this controls the maximum number of real roots to be isolated in any decomposition.\footnote{It is the maximum individual degree within the base for the double exponent in the complexity bound of CAD.}

\subsection{Plentiful useful side  conditions}

With variables frequently multiplying each other in a sentence's polynomials, a potentially large number of polynomial singularities might be relevant to a QE procedure.  However, the 45 sentences typically include sign conditions that in principle rule out the need to consider many such singularities, if the computational path of the algorithm can be designed to exploit this.

\subsection{Sparse occurrence of variables}

For each sentence we have formed a matrix of occurrences, with rows representing variables and columns representing polynomials.  A matrix entry is one if and only if the variable appears in the polynomial, and zero otherwise.  The average entry of all 45 occurrence matrices is 0.15.  The sparsity of the occurrence matrix indicates that we can limit the number of resultants that need to be calculated when applying projection operators.

\subsection{Tackling with existing QE technology}

Although the occurrence matrices are sparse, most of them are not sparse enough to make CAD projection operators, let alone construction of a full CAD, practical.  Take the scenario-analysis example from Section \ref{SUBSEC:Scenario} above, which is one of the simpler of the 45. It has sixteen polynomials, but application of Brown's projection operator \cite{Brown2001a} in alphabetical order of the variables encounters a lack of memory after eliminating just 5 of the 12 variables (prior to any real root isolation or stack construction).  Even when using a better variable-elimination order, we may conclude that Brown's projection operator produces thousands of polynomials.  

Although construction of a full CAD is not a practical method for deciding the 45 sentences, automatic QE is still possible in seconds: \textsc{Mathematica}'s \verb+Resolve+ function and \textsc{Redlog}'s \verb+rlqe+ function both decide most of them, presumably using methods other than CAD such as Virtual Substitution \cite{Weispfenning1988}.  
Many of the problems are fully quantified with existential quantifiers, making them SAT instances and in the scope of SMT\footnote{Satisfiability Modulo Theory (see for example \cite{AAB+16a})}-solvers which handle non linear real arithmetic, such as \textsc{Z3}, which decides all but two such problems.  We note that all implementations benefit from a good choice of variable ordering and that \textsc{Z3} has an additional sensitivity to the order in which constraints are presented in the logical formula.  Each implementation\footnote{We use version 11.2 of \textsc{Mathematica}, revision 4330 of \textsc{Redlog PSL}, and version 4.5.0 of \textsc{Z3}, all running on OSX.}  was advantageous for certain groups of examples. 

We note here that the polynomial list presented in Section 5.2 of \cite{BD07}, which has a projection set that is doubly exponential in the number of variables, is sparse like these economics examples and even simpler in terms of being linear in both their individual and total degree.  So like the economics examples, some Tarski formulas formed from the \cite{BD07} polynomials should be amenable to QE technology not relying on full CADs.

\section{QE incrementally by quantified variable}
\label{SEC:BlockQE}

Although existing technology can make progress with economics examples the unique structure suggests that a specialised QE technique may offer improved performance here.

In recent years there has been rapid development in new QE techniques, particularly adaptations of the CAD method, which offers computations more tailored to the particular application domain or logical structure. For example:
\begin{itemize}
\item making use of any Boolean structure in the input \cite{BDEMW16, EBD15}; 
\item local projection approaches \cite{Brown2013, Strzebonski2016}; 
\item non-uniform CADs (which relax the global cylindricity condition) \cite{Brown2015}; and their interaction with (Satisfiability Modulo Theory) SMT-solvers \cite{JdM12, AAB+16a}.
\item triangular decompositions via complex space \cite{CMXY09, BCDEMW14};
\end{itemize}
What we need for the examples in this collection is a framework which exploits the low degrees that individual variables occur in.  Also, the sparsity of the occurrence matrices described in the previous section suggests an incremental elimination approach.  

In particular, when removing a single existential quantifier, only a handful of the atoms in the formula contain the variable being eliminated, and often occurring in formulae with a familiar structure (meaning number, degree and type of constraints).  This suggests the economy of generating a small number of known QE results for formulae of generic structure which can be used repeatedly.  This allows for quantifiers to be removed incrementally without forming a full CAD or reflecting all of the variables in the formula.  We will demonstrate this by means of a working example.

\subsection{A working example}

We use an example adapted from the graduate-level microeconomics textbook \cite{JehleReny2011}.  The example asserts that any differentiable, quasi-concave, homogeneous, three-input production function with positive inputs and marginal products must also be a concave function.  In the Tarski framework, we have a twelve-variable sentence, which we have expressed simply with variable names $v_1, \dots, v_{12}$ to compress the presentation:

\begin{align*}
A &\equiv v_{1} v_{10}+v_{2} v_{7}+v_{3} v_{5}=0 \land v_{1} v_{11}+v_{2} v_{8}+v_{3} v_{7}=0
\\
&\quad 
\land v_{1} v_{12}+v_{10} v_{3}+v_{11} v_{2}=0 
\\
&\quad 
\land v_{1}>0 \land v_{2}>0 \land v_{3}>0 
\land v_{4}>0 \land v_{6}>0 \land	v_{9}>0 
\\
&\quad 
\land 2 v_{11} v_{6} v_{9}>v_{12} v_{6}^2+v_{8} v_{9}^2 
\\
&\quad 
\land 2 v_{10} v_{6} (v_{11} v_{4}+v_{7} v_{9})+v_{9} (2 v_{11} v_{4} v_{7}-2 v_{11} v_{5} v_{6}+v_{5} v_{8} v_{9})
\\
&\qquad 
	+ v_{12} \big(v_{4}^2 v_{8}-2 v_{4} v_{6} v_{7}+v_{5} v_{6}^2\big)>
\\
&\qquad \qquad 
	v_{10}^2 v_{6}^2+2 v_{10} v_{4} v_{8} v_{9}+v_{11}^2 v_{4}^2+v_{7}^2 v_{9}^2.
\\
\quad
\\
H &\equiv 
v_{12}\leq 0\land v_{5} \leq 0\land v_{8}\leq 0 
\\
&\qquad \land
	v_{12} v_{5}\geq v_{10}^2 \land
	v_{12} v_{8}\geq v_{11}^2\land
	v_{8} v_{5}\geq v_{7}^2
\\
&\qquad \land 
	v_{8} \left(v_{10}^2-v_{12} v_{5} \right)+v_{11}^2 v_{5}+v_{12} v_{7}^2\geq 2 v_{10} v_{11} v_{7} .
\end{align*}

In disjunctive normal form, $A \wedge \neg H $ is a disjunction of seven clauses.  Each atom of $H$, and therefore each clause of $A \wedge \neg H $ in DNF, corresponds to a principal minor of the production function's Hessian matrix.  So the clauses in the DNF are $A \land v_{12}>0, A \land v_5 >0 \land \dots$ in the same sequence used in $H$ above.  The existential quantifiers applied to $A \wedge \neg H $ can be distributed across the clauses to form seven smaller QE problems.

As of early 2018, \textsc{Z3} could not determine whether the Tarski formula $A \wedge \neg H $ is satisfiable: it was left running for several days without producing a result or error message.

Even if we take only the first clause of the disjunction (the one containing $v_{12} > 0$), the Tarski formula has twelve polynomials in twelve variables, and it is not practical to construct a full CAD for them.  For example, just three applications of Brown's projection operator to eliminate $\{v_{12} , v_{11}, v_{10}\}$, in that order, results in 200 unique polynomials with nine variables still remaining to eliminate.

\subsection{Two building blocks}

Our approach is to make use of generic QE results as building blocks which can be repeatedly used to solve a larger specific QE problems.  We need two such building blocks to resolve the working example.

We denote Block-A as the QE problem to eliminate the quantifier on the sole  common variable of the conjunction of linear equation, a linear inequality and a quadratic inequality:
\begin{align*}
\mbox{Block-A} &= \exists_x \big( a_{1,1}x +a_{1,0}=0 \land a_{2,1}x+a_{2,0}>0 
\\
&\qquad \qquad \land a_{3,2}x^2 + a_{3,1}x+a_{3,0}>0 \big).
\end{align*}
Here the coefficient subscripts indicate their position in the formula, so $a_{i,j}$ is the coefficient of $x^j$ in constraint $i$.  From now on we suppress the comma in the subscripts to save space.

We can use a QE implementation to find the equivalent quantifier free formulae for Block-A:
\begin{align*}
&\quad \left(a_{11}>0\land a_{32} a_ {10}^2+a_ {11}^2 a_{30}>a_{10} \
	a_{11} a_ {31}\land a_{10} a_{21}<a_{11} a_{20}\right) 
\\
&\lor
	\left(a_{10}
	a_{21}>a_{11} a_ {20}\land a_{32} a_ {10}^2+a_ {11}^2 \
	a_{30}>a_{10} a_{11} a_ {31}\land a_{11}<0\right)
\\
&\lor 
	\left(a_{10}=0\land a_{11}=0\land
	a_{20}>0\land a_{32}>0\right)
\\
&\lor
	\left(a_{10}=0\land \
	a_{11}=0\land a_{32}>0\land a_{21}\neq 0\right)
\\
&\lor
	\left(a_{10}=0\land a_{11}=0\land
	a_{32} a_ {20}^2+a_ {21}^2 a_{30}>a_{20} a_{21} a_ {31}\land \
	a_{21}\neq 0\right)
\\
&\lor
	\big(a_{10}=0\land a_{11}=0\land \
	a_{32} a_ {20}^2+a_ {21}^2
	a_{30}=a_{20} a_{21} a_{31}
\\
&\qquad \land
	2 a_{20} a_{32}<a_{21} a_{31}\land a_{21}\neq 0\big)
\\
&\lor 
	\left(a_{10}=0\land \
	a_{11}=0\land a_{21}>0\land
	a_{31}>0\land a_ {32}\geq 0\right)
\\
&\lor \left(a_{10}=0\land \
	a_{11}=0\land 2 a_{20} a_ {32}^2>a_{21} a_{31} a_ {32} \right.
\\
&\qquad\land \left. 4 \
	a_{30} a_ {32}^2>a_ {31}^2
	a_ {32}\land a_{32}\neq 0\right)
\\
&\lor \left(a_{10}=0\land \
	a_{11}=0\land a_ {32}\geq 0\land a_{21}<0\land a_{31}<0\right) 
\\
&\lor
	\left(a_{10}=0\land
	a_{11}=0\land a_{21}=0\land a_{20}>0\land a_{30}>0\land a_{32}\geq 0\right)
\\
&\lor 
	\left(a_{10}=0\land a_{11}=0\land a_{21}=0\
	\land a_{20}>0\land
	a_ {32}\geq 0\land a_{31}\neq 0\right)
\\
&\lor 
	\left(a_{10}=0\land \
	a_{11}=0\land a_{21}>0\land a_{30}>0\land a_ {31}\geq 0\land a_{32}\geq 0\right)
\\
&\lor
	\left(a_{10}=0\land a_{11}=0\land a_{30}>0\land a_ {32}\geq \
	0\land a_{21}<0\land a_ {31}\leq 0\right)
\end{align*}
While this looks lengthy, all but two clauses are of the form $$\left(a_{10}=0\land a_{11}=0\land\cdots\right),$$ and these clauses often simplify for specific values of the generic coefficients. The equality constraints in them also enable better processing \cite{BDEMW16}.

Despite the presence of the symbolic coefficients Block-A's full CAD is actually fairly simple to identify because $x$ is the only variable appearing in more than one polynomial.  Using Brown's projection operator (eliminating $x$ first and then the coefficients in alphabetical order), the entire projection set has only 11 polynomials!  We may substitute for the $a_{ij}$ to solve problems with specific polynomials without further use of QE technology.  

The second building block required to resolve the working example is a similar problem with an additional linear inequality present, which we denote Block-B:
\begin{align*}
\mbox{Block-B} &= \exists_x \big( a_{11}x+a_{10}=0 \land a_{21}x+a_{20}>0 
\\
&\qquad \qquad \land \, a_{31}x +a_{30}>0 
\land a_{42}x^2 +a_{41}x +a_{40}>0 \big).
\end{align*}
The quantifier free formula equivalent to Block-B is significantly larger and so omitted here (but can still be calculated instantly in \textsc{Redlog}).

\subsection{Building blocks applied repeatedly to eliminate quantifiers}

Take the first clause of $A \wedge \neg H$ from our working example in DNF, which is the one containing $v_{12} > 0$.  It is a conjunction of twelve atoms, only four of which contain $v_{12}$.  The four atoms are shown below, with $v_{12}$ highlighted in red.
\begin{align*}
&\quad \left(v_1 \textcolor{red}{v_{12}}+v_{10} v_3+v_{11} v_2=0\right) 
\\
&\land
	\left(2 v_{11} v_6 v_9>\textcolor{red}{v_{12}} v_6^2+v_8 v_9^2 \right) \land 
	\left(\textcolor{red}{v_{12}} >0\right) 
\\
&\land
	\big(2 v_{10} v_{11} v_4 v_6+2 v_{10} v_6 v_7 v9+2 v_{11} v_4 v_7 v_9
	- 2 v_{11} v5 v6 v9 
\\
&\qquad \quad
	+\textcolor{red}{v_{12}} v_4^2 v_8-2 \textcolor{red}{v_{12}} v_4 v_6 v_7+\textcolor{red}{v_{12}} v_5 v_6^2+v_5 v_8 v_9^2
\\
&\qquad 
	 > v_{10}^2 v_6^2+2 v_{10} v_4 v_8 v_9+v_{11}^2 v_4^2+v_7^2 v_9^2 \big)
\end{align*}

Notice that the conjunction of the four atoms is of the same format as Block-B, for example with $a_{11}$ interpreted as $v_1$.  We can therefore eliminate $v_{12}$ from this clause as follows: 
\begin{enumerate}[(i)] 
\item Take the quantifier free version of Block-B and substitute the generic coefficients for the polynomials in $\{v_{1}, \dots, v_{11}\}$ required for it to match the conjunction above.
\item Conjunct to this the remaining eight $v_{12}$-free atoms from the original clause.  
\end{enumerate}
The output will need to be simplified (we have used \textsc{Mathematica}).  The result of (i) is a DNF in which many clauses can be discarded.  For example, we had $a_{31}=1$ and $a_{30}=0$ and so many of the clauses in the quantifier free form of Block-B which require incompatible sign conditions on $a_{31}$ and $a_{30}$ must be \texttt{false} and so can be discarded directly from that disjunction, leaving in this case a 16 clause DNF.  

Further simplification can then be applied after conjunction with the constraints added by (ii).  In this example one of those constraints was $v_1>0$ which is in incompatible with the $v_1=0$ appearing in 15 of the 16 clauses.  Thus here the output simplifies to a single conjunction, but more generally we would need to form a DNF and continue distributing quantifiers over clauses and solving sub-problems.

For our example, after simplification the resulting formula is a conjunction of eleven atoms, four of which contain $v_{11}$:
\begin{align*}
&\quad \left(v_{1} \textcolor{red}{v_{11}}+v_{2} v_{8}+v_{3} v_{7}=0\right) \land 
	\left(v_{10} v_{3}+\textcolor{red}{v_{11}} v_{2}<0\right) 
\\
&\land
	\left(2 v_{1} \textcolor{red}{v_{11}} v_{9}+v_{10} v_{3} v_{6}+\textcolor{red}{v_{11}} v_{2} v_{6})>v_{1} v_{8} v_{9}^2 \right) 
\\
&\land 
	\bigg(v_{1} \Big(v_{1} \big(v_{10}^2 v_{6}^2-2
	\textcolor{red}{v_{11}} (v_{10} v_{4} v_{6}+v_{4} v_{7} v_{9}-v_{5} v_{6} v_{9})
\\ &\qquad
	+v_{10} (2 v_{4} v_{8} v_{9}-2
	v_{6} v_{7} v_{9})+\textcolor{red}{v_{11}}^2 v_{4}^2+v_{9}^2 (v_{7}^2-v_{5} v_{8})\big)
\\ &\qquad
	+(v_{10} v_{3}+\textcolor{red}{v_{11}}
	v_{2}) (v_{4} (v_{4} v_{8}-2 v_{6} v_{7})+v_{5} v_{6}^2 )\Big)<0\bigg)
\end{align*}
We can once again apply Block-B; substitute for the generic coefficients; conjunct to this the remaining 7 atoms; and simplify.  We repeat this two further more times to eliminate $v_{10}$ and $v_{8}$.  In each of these three eliminations the simplification leaves only one clause of the DNF.  

The result at this stage is a conjunction of nine atoms, three of which contain $v_{2}$.  One of the atoms is an equality that is linear in $v_{2}$, another is a strictly linear inequality, and the third is a strict quadratic inequality.  We hence this time apply Block-A to eliminate $v_{2}$.  Upon simplification the resulting quantifier free formula is shown to be \texttt{false} (all clauses of the resulting DNF had incompatible constraints).  

We are not yet finished as we have only studied the first of the 7 clauses in the DNF of $A \wedge \neg H$.  So far we have shown there is no counterexample satisfying $v_{12}$ > 0.  The second clause of the original disjunction (the one containing $v5 > 0$) can be decided by applying Block-A four times, using the same variable-elimination order as above, to arrive at a contradiction in a similar manner.

The third and sixth disjunction clauses, respectively, are decided by applying the building blocks (with the same variable order) as follows:
\begin{itemize}
\item Block-A, Block-A, Block-A, Block-B;
\item Block-A, Block-A, Block-A.
\end{itemize}
The final three disjunction clauses are decided by eliminating $v_{12}$, followed by $v_{8}$, and (for the two in which is is necessary) then $v_{7}$.  The building-block applications are:
\begin{itemize}
\item Block-B, Block-A, Block-A;
\item Block-B, Block-B, Block-B;
\item Block-B, Block-B.
\end{itemize}
As with the first two clauses, all five of these clauses result in a contradiction.  $A \wedge \neg H $, which is a disjunction of the seven clauses, is therefore \texttt{false}: any differentiable, quasi-concave, homogeneous, three-input production function with positive inputs and marginal products must also be a concave function.  

In summary, for this example, seven twelve-variable existential sentences are decided with repeated application of two simple prototype QE problems.

\subsection{General algorithm}

We present as Algorithm \ref{alg:ByBlocks} overleaf, the general high level approach we propose, as described above for the working example.

\begin{algorithm}[t!]
\caption{QE Incrementally by variable using blocks}
\label{alg:ByBlocks}
\DontPrintSemicolon
\SetKwInOut{Input}{Input}\SetKwInOut{Output}{Output}
\Input{A formula quantified in variables $v_1,\ldots,v_n$.}
\Output{Either \texttt{true}, \texttt{false}, or an equivalent quantifier free formula in any free variables.}
\BlankLine
Arrange the formula so that all quantifiers are existential over a disjunctive normal form. \;
Let $\mathcal{C}$ be the set of clauses in the input formula.\;
Prepend to each $c \in \mathcal{C}$ the quantifiers (and corresponding quantified variables) of the formula. \;
\While{there is a quantified $c \in \mathcal{C}$:}{
Choose a quantified variable $v$.\;
Set $c_1$ to be the conjunction of atoms in $c$ which contain $v$ and $c_2$ to be the conjunction of the remaining.\;  
View $c_1$ as univariate in one of the quantified variables and observe the structure (degree, number and type of constraints).\;
Select the corresponding generic QE result from stored results (or if not stored then generate).\;
Substitute generic coefficients for the corresponding polynomial coefficients of $v$ from $c_1$ to form a DNF $\phi$.\;
\For{each clause $d$ in $\phi$}{
Form $e = d \land c_2$.\;
\If{$e$ simplifies to \texttt{false}}{
\texttt{continue} to next entry in $\phi$.\;
}
\If{$e$ simplifies to \texttt{true}}{
\Return \texttt{true}.\;
}
Otherwise prepend the remaining quantified variables from $c$ to $e$ and add to $\mathcal{C}$.
}
}
\Return The disjunction of unquantified constraints in $\mathcal{C}$.
\end{algorithm}

The correctness of the algorithm is not hard to see.  We recall that existential quantifiers distribute over disjunctions so the disjunction of solutions to the sub-problems formed in $\mathcal{C}$ by Step 4 is a solution to the input problem.  

The efficiency of the algorithm is dependent on Step 8 requiring only a small number of previously generated QE results, all univariate with low degree; and the simplification of $e$ leading to a Boolean in most cases of the inner loop.  These conditions seem to often be the case for examples from economics.

We note that Algorithm \ref{alg:ByBlocks} covers the case where the evaluation of a clause could involve a disjunction, not present for our worked example.  The order of discussion of the worked examples suggests to eliminate all variables from one parent clause before addressing the next, but as Algorithm \ref{alg:ByBlocks} demonstrates this is not necessary.  Exhausting one clause would be a good strategy if the clause evaluated to \texttt{true} as computation could cease early but otherwise offers no particular benefit.

\section{Recursive generation of QE output for additional low degree formulae}
\label{SEC:BuildBlock}

The resolution of a large problem by only a few generic QE results as in the last examples appears to be common within the examples generated by economics theory.  As described above, none of our example problems involve variables which appear with degree greater than three.  However, the number of (low degree) polynomial constraints involved can be far greater than in the preceding example.  Thus we need to generate further generic QE results other than Block-A and Block-B.  In this section we explore how one might derive these with a minimal use of QE technology.

\subsection{Result for combining linear inequalities}

In this section let $F(x, M)$ be a conjunction of $M$ inequalities each linear in the variable $x$ with coefficients polynomial in the variables $\bm{z} = (z_1, \dots, z_n)$. So
\begin{align*}
F(x,M) &= \bigwedge_{i=1}^{M} P_i \theta_i 0 \\
&= P_1(x) \theta_1 0 \land P_2(x) \theta_2 0 \land \dots \land P_M(x) \theta_M 0
\end{align*}
where each $P_i(x) = a_i(\bm{z})x + b_i(\bm{z})$ and each $\theta_i \in \{>, \geq\}$.

\begin{proposition*}
With the notation above we have the following logical equivalence:
\begin{equation*}
\exists_x F(x,M + 1) = \underbrace{\exists_xF(x,M)}_{C_0} \land
\bigwedge_{i=1}^{M}  
\underbrace{\exists_x P_{M+1}(x) \theta_{M+1} 0 \land P_i(x) \theta_i 0}_{C_i}.
\end{equation*}
\end{proposition*}

\begin{proof}
Clearly if the left hand side is true then the right must be also, since satisfying them all together would imply satisfying all but one ($C_0$) as well as any pair (the other $C_i$).  

The opposite direction is less obvious.  If any $C_i$ is false then the right overall becomes false; as would the left since it requires a superset of the unsatisfiable condition.  So from now on we assume the truth of each $C_i$, and we aim to show that the $x$'s asserted to exist by each $C_i$ could be taken equal.

We will assume that $\forall i, a_i(\bm{z}) \ne 0$. If this where not the case for any $a_i(\bm{z})$ then $P_i(x)\theta_i 0$ becomes $b_i(\bm{z}) \theta_i 0$, which for fixed $\bm{z}$ is either true immediately (effectively reducing $M$ by one) or false immediately (ruled out above).  

For simplicity let us prove first the case $M=2$, i.e.
\begin{eqnarray*}
&\exists_x P_1(x) \theta_1 0 \land P_2(x) \theta_2 0 \land P_3(x) \theta_3 0 =\underbrace{\exists_x P_1(x)\theta_10\land P_2(x)\theta_20}_{C_0}\cr&
\underbrace{\exists_xP_3(x) \theta_3 0\land P_1(x)\theta_10}_{C_1}
\quad \land \quad \underbrace{\exists_x P_3(x) \theta_3 0\land P_2(x)\theta_20}_{C_2}.
\end{eqnarray*}
Then each
$P_i(x)\theta_i0$ may be written as 
$x \phi_i(\bm{z})c_i(\bm{z})$
where $c_i(\bm{z})=-b_i(\bm{z})/a_i(\bm{z})$ and  $\phi_i(\bm{z})$ is either $\theta_i$ or the inequality of the opposite direction, as in the table below.
\begin{center}
\begin{tabular}{c|cccc}
$\theta_i$   & $>$ & $<$ & $\geq$ & $\leq$  \\
$\phi_i$ & $\in\{>, <\}$ & $\in\{>, <\}$ & $\in\{\geq, \leq\}$ & $\in\{\geq, \leq\}$
\end{tabular}
\end{center}
Consider $C_0$ which we write as  
\begin{equation}\label{eq:1}
x\phi_1c_1\land x\phi_2c_2
\end{equation} 
(dropping the $(\bm{z})$ since we are working over a fixed $\bm{z}$). 
If $\phi_1$ and $\phi_2$ are in opposite directions then (\ref{eq:1}) defines a finite interval: either $\I{c_1,c_2}$ or $\I{c_1,c_2}$ (the parentheses notation implies that we do not know whether the interval is open or closed).
If  $\phi_1$ and $\phi_2$ are in the same direction then we have a semi-infinite interval bounded at one end by either $c_1$ or $c_2$.  

So the truth of $C_0$ asserts the existence of an $x$ to satisfy the first two inequalities within an interval $\I{\ell_1,\ell_2}$ where $\ell_1 \in \{c_1, c_2, -\infty$ and $\ell_2 \in \{c_1, c_2, \infty\}$.

We now need to show that there is an $x \in \I{\ell_1,\ell_2}$ which satisfies $P_3 \theta_3 0$ also.  We must have one of the following two cases.
\begin{description}
\item[$\phi_3$ is $>$ or $\geq$:] Then the third inequality requires $x \geq c_3$. So the interval of $x$ to satisfy all inequalities is $\I{max(\ell_1,c_3), \ell_2}$.  If $\ell_2$ is $\infty$ then the interval is clearly non-empty.  Otherwise the interval is bounded at either end by one of $\{c_1, c_2, c_3\}$ and the non-emptiness is implied by the truth of the particular $C_i$ generated by the two inequalities whose indices appear in those bounds.    
\item[$\phi_3$ is $<$ or $\leq$:] Then the third inequality requires $x \leq c_3$ and we need to find $x \in \I{\ell_1, min(\ell_2,c_3)}$ which is concluded non-empty similarly.
\end{description}
This concludes the proof for $M=2$.  With general $M$ the argument is similar:  for a fixed $\bm{z}$ each $P_i \theta_i 0$ defines a semi-infinite interval; the intersection of the first $M$ gives a non-empty interval from the truth of $C_0$; the final inequality defines a last semi-infinite interval which has non-zero intersection by reference to the truth of the appropriate $C_i$.
\end{proof}

The proposition allows us to apply QE techniques once for a formulae with just two linear inequalities, and then recursively derive the formulae for as many linear inequalities as needed.  Each step simply requires a renaming to the original QE output.

For example, consider the conjunction of two strict inequalities:
\[
\exists x \left(0 < a_{11}x+a_{10} \land 0 < a_{21}x+a_{20} \right)
\]
Any QE implementation would derive the equivalent quantifier free formula
\begin{align}
&(a_{21} < 0 \land a_{11} < 0) 
\lor \, (a_{21} < 0 \land 0 < a_{11} \land a_{10}a_{21}-a_{20}a_{11} < 0) \nonumber \\
&\lor \, (a_{21} < 0 \land a_{11} = 0 \land 0 < a_{10}) 
\lor \, (a_{21} = 0 \land a_{11} < 0 \land 0 < a_{20})  \nonumber \\
&\lor \, (a_{21} = 0 \land a_{11} = 0 \land 0 < a_{20} \land 0 < a_{10}) 
\lor \, (0 < a_{21} \land 0 < a_{11}) \nonumber \\
&\lor \, (a_{21} = 0 \land 0 < a_{11} \land 0 < a_{20}) 
\lor \, (0 < a_{21} \land a_{11} = 0 \land 0 < a_{10}) \nonumber \\
&\lor \, (0 < a_{21} \land a_{11} < 0 \land 0 < a_{10}a_{21}-a_{20}a_{11}).
\label{QE2Strict}
\end{align}
Then if we instead consider a problem with a third inequality
\[
\exists x (0 < a_{11}x+a_{10} \land 0 < a_{21}x+a_{20} \land 0< a_{31}x+a_{30}
\]
we can find an equivalent quantifier free formula using a simple relabelling:
\begin{equation}
(\ref{QE2Strict}) + (\ref{QE2Strict})\big|_{a_{21}=a_{31}, a_{20}=a_{30}}
+ (\ref{QE2Strict})\big|_{a_{11}=a_{31}, a_{10}=a_{30}}
\label{QE3Strict}
\end{equation}

\subsection{Processing the output}

The direct evaluation (\ref{QE3Strict}) would produce a formula quite different to the output of QE directly (although the two must be equivalent).  In particular, QE technology usually gives output in disjunctive normal form (DNF), certainly if it is based on cell or case decomposition techniques such as CAD or VS.  
It is likely that to make use of formulae like (\ref{QE3Strict}) a conversion to DNF will be needed.  For example, once in DNF we can apply a simplifier and split into sub-problems for further analysis.  

It is well known that such a conversion can lead to an exponential increase in formula size.  Indeed, the 27 atoms of (\ref{QE3Strict}) become a disjunction of 505 conjunctive clauses under that standard transformation.  
However, via some simple simplification we can conclude the majority of these false (as they contain mutually exclusive constraints on the sign of a variable).  In fact for (\ref{QE3Strict}) we can remove all but 27 clauses this way.  However, the simplification procedure still needed to process the 505.  The 4, 5 and 6 inequality variants have 81, 243 and 729 clauses after simplification so the exponential growth is still apparent - and the final one required tens of thousands of intermediate expressions to be analysed.  Nevertheless, generating the QE output this way is still quicker than doing so directly.

%
%
%

\section{Conclusions and Future Work}

We have demonstrated that there are a wide variety of examples from economics that may be tackled with QE.  They tend to have a larger number of variables than those typical in the QE literature, but also other simplifying characteristics which mean their automatic resolution is possible but would benefit from specialised techniques such as those sketched in Section \ref{SEC:BlockQE}.

\vspace*{0.1in}

There is much scope for future development of the ideas presented here.  Perhaps most useful would be further results like the proposition in Section \ref{SEC:BuildBlock} that allow for the generation of QE building blocks with large numbers of low degree univariate constraints.  These are currently under investigation.

As with many QE techniques there are questions of computation path and algorithms settings that do not call into question the correctness of the output but can have a significant affect on computation time and memory resources.  The best known is the ordering of variables for elimination.  It is well known that such an ordering can be crucial to the tractability of CAD\footnote{Even to the point of changing the complexity from constant to exponential \cite{BD07}} and it will have a similar importance to the approach of Section \ref{SEC:BlockQE}.  Existing heuristics may be applicable \cite{DSS04, BDEW13}, and we note that when considering Algorithm \ref{alg:ByBlocks} it would be permissible to use different orderings for each sub-problem in $\mathcal{C}$ (indeed, this was the approach described for the working example). In addition the order in which those sub-problems are studied and whether to study entire sub-problems in one go are independent areas for optimisation.  Strategies may be developed based on human designed heuristics but also perhaps machine learned techniques as used in \cite{HEWDPB14, HEDP16, KIMA16}.

\vspace*{0.1in}

We note that economics has the potential to be a large application domain for QE technology.  Within education QE tools tailored to economics could be utilised throughout the curriculum to remove much arduous computation and allow greater intuitive experimentation with the ideas.  Within research tools could allow quicker development and testing of ideas with lower risk of error. In the other direction, economics offers examples and inspiration for combined/tailored symbolic algorithms.\footnote{As has recently been the case with biology \cite{BDEEGGHKRSW17, EEGRSW17}.}  Both communities would prosper from further collaboration.
}


\bibliographystyle{plain}
\bibliography{CAD}

\end{document}